 \documentclass[final,5p,times,twocolumn]{elsarticle}

\usepackage{amsmath,amsfonts,amssymb,mathtools}
\usepackage{bbm}
\usepackage{graphicx}
\usepackage{exscale}

\newtheorem{thm}{Theorem}
\newtheorem{lem}[thm]{Lemma}
\newtheorem{rem}[thm]{Remark}
\newtheorem{exmp}{Example}

\newcommand{\E}[1]{\mathop{{\rm \bf E}\left\{#1\right\}}\nolimits}
\DeclareMathAlphabet\mathbfcal{OMS}{cmsy}{b}{n}

\usepackage{clrscode}
\begin{document}

\begin{frontmatter}
\title{Square-Root Algorithms for Maximum Correntropy Estimation of Linear Discrete-Time Systems in Presence of non-Gaussian Noise}

\author[CEMAT]{M.~V.~Kulikova\corref{cor}} \ead{maria.kulikova@ist.utl.pt} \cortext[cor]{Corresponding
author.}

%(Center for Computational and Stochastic Mathematics)
\address[CEMAT]{CEMAT (Center for Computational and Stochastic Mathematics), Instituto Superior T\'{e}cnico, Universidade de Lisboa, \\ Av. Rovisco Pais 1,  1049-001 Lisboa, Portugal}

\begin{abstract}
Recent developments in the realm of state estimation of stochastic dynamic systems in the presence of non-Gaussian noise have induced a new methodology called  the maximum correntropy filtering. The filters designed under the maximum correntropy criterion (MCC) utilize a similarity measure (or {\it correntropy}) between two random variables as a cost function. They are shown to improve the estimators' robustness against outliers or impulsive noises. In this paper we explore the numerical stability of {\it linear} filtering technique proposed recently under the MCC approach. The resulted estimator is called the maximum correntropy criterion Kalman filter (MCC-KF). The purpose of this study is two-fold. First,  the previously derived MCC-KF equations are revised and the related Kalman-like equality conditions are proved. Based on this theoretical finding, we improve the MCC-KF technique in the sense that the new method possesses a better estimation quality  with the reduced computational cost compared with the previously proposed MCC-KF variant. Second, we devise some square-root implementations for the newly-designed improved estimator. The square-root algorithms are well known to be inherently more stable than the {\it conventional} Kalman-like implementations, which process the full error covariance matrix in each iteration step of the filter. Additionally, following the latest achievements in the KF community, all square-root algorithms are formulated here in the so-called {\it array} form. It implies the use of orthogonal transformations for recursive update of the required filtering quantities and, thereby, no loss of accuracy is incurred. Apart from the numerical stability benefits, the array form also makes the modern Kalman-like filters better suited to parallel implementation and to very large scale integration (VLSI) implementation. All the MCC-KF variants developed in this paper are demonstrated to outperform the previously proposed MCC-KF version in two numerical examples.
\end{abstract}

\begin{keyword}
Maximum correntropy criterion, Kalman filter, square-root filtering, robust estimation.
\end{keyword}

\end{frontmatter}

\section{Introduction}

In the past few years, the study of filtering techniques under the maximum  correntropy criterion (MCC) has become an important aspect of a hidden state estimation of stochastic dynamic systems in the presence of non-Gaussian noise~\cite{liu2007correntropy,gunduz2009correntropy,principe2010information,Cinar2012}. The MCC methodology implies that a statistical metric of a {\it similarity} between two random variables (or {\it correntropy}) is used as a cost function (or performance index) for designing the corresponding estimation method. The resulted MCC filters have become the methods of choice in signal processing and machine learning due to its robustness against outliers or impulsive noises compared to the classical Kalman filtering (KF); e.g., see the discussion in~\cite{xu2008pitch,he2011maximum,chen2014steady,chen2015convergence,ma2015maximum} and many others.

Being a {\it linear} estimator, the KF is an attractive and simple technique that requires only the computation of mean and covariance for constructing the optimal estimate of unknown dynamic state under the minimum mean square (MMS) criterion. For Gaussian systems, this estimate is optimal, i.e. the KF reduces to  an MMS estimate rather than a {\it linear} MMS estimate. It is clear that in non-Gaussian setting, the classical KF exhibits sub-optimal behavior only. Due to this fact, there was a need for a new estimator that improves the KF robustness  against outliers or impulsive noises.

For linear non-Gaussian state-space models, the robust maximum correntropy Kalman filter (MCKF) and the maximum correntropy criterion Kalman filter  (MCC-KF) have been recently developed in~\cite{Chen2017,liu2017} and~\cite{izanloo2016kalman}, respectively.  As all Kalman-like filtering algorithms, they compute the first two moments (i.e. the mean and the covariance) for constructing the optimal estimate. However, in contrast to the classical KF, these recent developments utilize the robust MCC as the optimality criterion, instead of using the MMS cost function. As a result, the new filters are shown to outperform the classical KF and several nonlinear Kalman-like filtering techniques in the presence of non-Gaussian uncertainties in the state-space models. Nevertheless, little attention is paid to  numerical stability of the Kalman-like filters developed under the MCC strategy, although the classical KF is widely known to suffer from the influence of  roundoff errors, severely; see~\cite{VerhaegenDooren1986,Verhaegen1989}. Our research has tended to focus on the MCC-KF technique and the design of its numerically stable square-root implementations.
%Consequently, any Kalman-like technique inherits this drawback. In this paper we explore the MCC-KF method developed  recently in~\cite{izanloo2016kalman}.
%We stress that for the recently derived MCKF method, a numerically stable square-root approach is discussed in~\cite{Chen2017} that is not the case for the %MCC-KF designed in~\cite{izanloo2016kalman}. Here, we would like to fill in this gap and, hence, we revise and explore the MCC-KF technique.

The purpose of this paper is two-fold. First, we revise the previously derived MCC-KF equations and prove the related Kalman-like equality conditions. Based on this theoretical finding, we improve the previously proposed MCC-KF algorithm in the sense that the new filter (abbreviated as IMCC-KF) possesses a better estimation quality with the reduced computational cost. Second, we devise some square-root IMCC-KF implementations grounded in numerically robust orthogonal transformations. The square-root strategy is the most popular approach used for enhancing the filter numerical robustness; see~\cite{KaminskiBryson1971,Bierman1977,Sayed1994,ParkKailath1995} {\it etc}. It implies the Cholesky decomposition of error covariance matrix and, then, recursive re-calculation of its Cholesky factors instead of using full matrix. Following the latest achievements in the KF community, all square-root algorithms are formulated here in the so-called {\it array} form. This means that numerically stable orthogonal transformations are used as far as possible for updating the Cholesky factors in each iteration step. This provides a more reliable estimation procedure as explained in~\cite[Chapter~12]{KailathSayed2000}. Apart from numerical advantages, array Kalman-like algorithms are easier to implement than the explicit filter equations, because all required quantities are simply read off from the corresponding filter post-arrays. As mentioned in~\cite{ParkKailath1995}, this makes the modern KF-like algorithms better suited to parallel implementation and to very large scale integration (VLSI) implementation. Finally, all algorithms developed in this paper are demonstrated to outperform the previously proposed MCC-KF technique in two numerical examples.

\section{Maximum Correntropy Criterion Kalman Filter} \label{sec:state}
Consider the state-space equations
  \begin{align}
   x_{k} = & F_{k-1} x_{k-1}+ G_{k-1} w_{k-1}, \quad k \ge 1, \label{eq:st:1} \\
    z_k  = & H_k x_k+v_k   \label{eq:st:2}
  \end{align}
where  $x_k \in \mathbb R^n$  and $z_k \in \mathbb R^m$ are the unknown dynamic state and the observable measurement vector,
respectively. The processes $\{ w_k \}$ and $\{ v_k\}$ are zero-mean, white, uncorrelated, and have known covariance matrices $Q_k$ and $R_k$, respectively. They are also uncorrelated with the initial state $x_0$, which has the mean $\bar x_0$ and the covariance matrix $\Pi_0$.

The KF associated with state-space model~\eqref{eq:st:1}, \eqref{eq:st:2} yields the linear MMS estimate, $\hat x_{k|k}$, of the unknown dynamic state, given the available measurements $\{z_1,\ldots, z_{k}\}$. To improve the filter estimation quality in the presence of non-Gaussian noise, the MCC optimality criterion can be used instead of the MMS cost function for deriving the corresponding filtering equations. The performance index to be optimized under the MCC (with Gaussian kernel) approach is given as follows~\cite{Cinar2012,izanloo2016kalman}:
\begin{equation*}
J_m (x_k)  = G_{\sigma}\left(\|z_k - H_kx_{k}\|\right)  + G_{\sigma}\left(\|x_{k} - F_{k-1}x_{k-1}\|\right)  \label{eq:PI}
\end{equation*}
where $G_{\sigma}(\|x_k-y_k\|) = {\rm exp}\left\{ -{\|x_k -y_k\|^2}/{(2\sigma^2)}\right\}$, and $\sigma > 0$ is the kernel size or bandwidth.
%As mentioned in~\cite[p.~2]{Cinar2012} and~\cite[p.~501]{izanloo2016kalman}: ``One nice property of correntropy is that it is positive and bounded and with Gaussian kernel it reaches its maximum if and only if $X=Y$''.

Minimization of the objective function $J_m$ with respect to $x_k$ implies $\partial J_m/ \partial  x_k = 0$ and yields the equation~\cite{Cinar2012}:
\begin{equation}
(x_{k} - F_{k-1}x_{k-1}) = \frac{G_{\sigma}\left(\|z_k - H_kx_{k}\|\right) }{G_{\sigma}\left(\|x_{k} - F_{k-1}x_{k-1}\|\right)}H_k^T  (z_k - H_kx_{k}).
\label{eq:grad:cond}
\end{equation}

We note that the best estimate for state vector $x_{k-1}$ at time point $k-1$ is {\it a posteriori} estimate $\hat x_{k-1|k-1}$. Hence, from~\eqref{eq:grad:cond} one obtains the following nonlinear equation, which needs to be solved with respect to $x_k$:
\begin{equation}
 x_{k}\! = \! F_{k-1}\hat x_{k-1|k-1} \!+\! \frac{G_{\sigma}\left(\|z_k - H_kx_{k}\|\right) }{G_{\sigma}\left(\|x_{k} - F_{k-1}\hat x_{k-1|k-1}\|\right)}H_k^T\!(z_k - H_kx_{k}). \!\! \label{eq:nonlinear}
\end{equation}

The fixed point correntropy filter developed in~\cite{Cinar2012} and the MCC-KF method proposed in~\cite{izanloo2016kalman} suggest to use a fixed point rule for solving the mentioned nonlinear equation with initial approximation $x_k^{(0)} = \hat x_{k|k-1}$ at the right-hand side of~\eqref{eq:nonlinear}. Besides, both techniques imply only one iteration of the fixed point rule and, hence, by substituting $x_k \approx \hat x_{k|k-1}$ into the right-hand side of formula~\eqref{eq:nonlinear} we obtain the following recursion
\begin{equation*}
\hat x_{k|k} = F_{k-1}\hat x_{k-1|k-1} + \frac{G_{\sigma}\left(\|z_k - H_k\hat x_{k|k-1}\|\right) }{G_{\sigma}\left(\|\hat x_{k|k-1} - F_{k-1}\hat x_{k-1|k-1}\|\right)}H_k^T(z_k - H_k\hat x_{k|k-1}). \label{eq:nonlinear:2}
\end{equation*}

Next, the MCC-KF method designed in~\cite{izanloo2016kalman} integrates the KF minimum-variance estimation with the maximum correntropy filtering.  In particular, the cited paper utilizes the norm $\|\cdot\|_{R_k^{-1}}$ induced by the inverse measurement covariance matrix $R_k^{-1}$ in the numerator and the norm $\|\cdot\|_{P_{k|k-1}^{-1}}$ induced by the inverse predicted process covariance matrix $P_{k|k-1}^{-1}$ in the denominator of the recursion above. Thus, the MCC-KF is given as follows; see Algorithm~2 in~\cite{izanloo2016kalman}:

{\bf Initialization:}
\begin{align}
\hat x_{0|0} & =\E{x_0}, & P_{0|0} & = \E{ (x_{0}-\hat x_{0|0})(x_{0}-\hat x_{0|0})^T}. \label{mcc:initial}
\end{align}

{\bf Prior estimation:}
\begin{align}
\hat x_{k|k-1} & = F_{k-1}\hat x_{k-1|k-1}, \label{mcc:p:X} \\
P_{k|k-1} & = F_{k-1}P_{k-1|k-1}F_{k-1}^T+G_{k-1}Q_{k-1}G_{k-1}^T. \label{mcc:p:P}
\end{align}

{\bf Posterior estimation:}
\begin{align}
L_{k} & = \frac{G_{\sigma}\left( \| z_k - H_k\hat x_{k|k-1} \|_{R_k^{-1}} \right)}{G_{\sigma}\left( \| \hat x_{k|k-1} - F_{k-1} \hat x_{k-1|k-1}\|_{P^{-1}_{k|k-1}}\right)}, \label{mcc:f:L} \\
K^L_{k} & = (P_{k|k-1}^{-1}+L_kH_k^TR_k^{-1}H_k)^{-1}L_kH^T_kR_k^{-1}, \label{mcc:f:K} \\
\hat x_{k|k} & =    \hat x_{k|k-1}+K^L_{k}(z_k-H_k\hat x_{k|k-1}),   \label{mcc:f:X} \\
P_{k|k}  & = (I - K^L_{k}H_k)P_{k|k-1}(I - K^L_{k}H_k)^T\!+\!K^L_kR_k(K^L_k)^T. \label{mcc:f:P}
\end{align}

In the equations above, we use the new notation $K_k^L$ for the gain matrix $(P_{k|k-1}^{-1}+L_kH_k^TR_k^{-1}H_k)^{-1}L_kH^T_kR_k^{-1}$ appeared in Algorithm~2 in~\cite{izanloo2016kalman}, emphasizing the dependence of this quantity on the scalar $L_k$. This also helps us to distinguish this matrix from the classical KF feedback gain in the rest of our paper.

The readers are referred to~\cite{izanloo2016kalman} for a detailed derivation and properties of the MCC-KF estimator under consideration. In the cited paper, the MCC-KF is shown to outperform the classical KF,  the fixed point correntropy filter from~\cite{Cinar2012} and several nonlinear filtering techniques when the non-Gaussian uncertainties arise in stochastic system~\eqref{eq:st:1}, \eqref{eq:st:2}. %In contrast to the classical KF, the measurements with very large outliers do not affect the state estimate in the MCC-KF.

It is  worth noting here that because of utilizing only one iteration of a fixed point rule for solving the underlying nonlinear equation~\eqref{eq:nonlinear}, we have $G_{\sigma}\left(\|\hat x_{k|k-1} - F_{k-1}\hat x_{k-1|k-1}\|\right)=G_{\sigma}\left(\| 0 \|\right)=1$ since $\hat x_{k|k-1} = F_{k-1}\hat x_{k-1|k-1}$. Hence, both methods in~\cite{Cinar2012,izanloo2016kalman} can be simplified since the denominator in~\eqref{mcc:f:L} is equal to 1.
For further iterates, this is not the case and the difference might be considerable. For this reason, the general form of~\eqref{mcc:f:L} is used in this paper.

The kernel size $\sigma$ plays a significant role in the behavior of any correntropy filter. For instance, the MCKF developed in~\cite{Chen2017} was shown to be reduced to the standard KF as $\sigma \to \infty$. Here, we follow the adaptive strategy suggested in~\cite{izanloo2016kalman} and implemented in~\cite{izanloo2016codes} for choosing $\sigma$ (i.e.   $\sigma = \| z_k - H_k\hat x_{k|k-1} \|_{R_k^{-1}}$ in each iteration step) in order to provide a fair comparative study with the earlier published MCC-KF method. This strategy is also motivated by a case study presented in~\cite{Cinar2012}. We stress that the problem of optimal kernel size selection is beyond the scope of this paper.

In this paper, we explain how the MCC-KF estimation quality can be further enhanced. The new improved filter is based on the Kalman-like equations proved in  Section~\ref{sec:3}. Additionally, we derive two numerically stable square-root implementations, which are the main purpose in the present study.

\section{New Improved Maximum Correntropy Criterion KF} \label{sec:3}

The previously proposed MCC-KF algorithm was shown to be coincident with the classical KF when $L_k = 1$; see~\cite{izanloo2016kalman}. To begin designing a new estimator, we first note that for the classical KF the following formulas hold~\cite[p.~128-129]{simon2006optimal}:
\begin{align}
\!\!K_{k}   & = P_{k|k}H_k^TR_k^{-1} \label{kf:K:eq2} \\
& = P_{k|k-1}H_k^TR_{e,k}^{-1}, \; R_{e,k} = H_kP_{k|k-1}H_k^T+R_k. \label{kf:K:eq1} \\
\!\!P_{k|k}  & = \left( P_{k|k-1}^{-1}+H^T_kR_k^{-1}H_k\right)^{-1} \label{kf:P:eq1} \\
  & = (I - K_{k}H_k)P_{k|k-1} \label{kf:P:eq2} \\
  & = (I - K_{k}H_k)P_{k|k-1}(I - K_{k}H_k)^T+K_kR_kK_k^T \label{kf:P:eq3}
\end{align}
The important property of the KF for Gaussian state-space models~\eqref{eq:st:1}, \eqref{eq:st:2} is $e_k \sim {\cal N}\left(0, R_{e,k}\right)$ where $e_k = z_k-H_k\hat x_{k|k-1}$ are called innovations (or residuals) of the discrete-time KF.

%Although the matrices $F$, $B$, $G$, $H$, $Q$ and $R$ are constant, the KF is still time-varying. It can be shown that
%observability is sufficient for convergence of the KF:
%if the pair $(F, H)$ is observable, then the Riccati recursion for $P_{k+1|k}$ converges to steady-state value $\lim \limits_{k \to \infty} P_{k+1|k} = %P_{\infty}$; see~\cite{Moore2005,Simon2006} for a detailed explanation.  We may also remark that if $P_{k+1|k}$ converges, then so do $K_{p,k}$ and $R_{e,k}$, %i.e. we get a stationary error distribution: $\lim \limits_{k \to \infty} e_k \to {\cal N } (0, R_{e,\infty})$ with finite variance $\lim \limits_{k \to \infty} %R_{e,k} = R_{e,\infty}$. Hence, the uncertainty in the state estimates remains bounded for all states. Next, if the state-space model is reachable, then the KF %reaches a steady-state value. Thus, under the assumptions that the pair $(F, H)$ is observable and the pair $(F, GQ^{1/2})$\footnote{The matrix $Q^{1/2}$ is the %square-root of $Q$, i.e. $Q=Q^{1/2}Q^{T/2}$.}   is reachable, the KF converges to the unique value and this stationary filter is asymptotically optimal. These %assumptions are also standard for most subspace identification algorithms~\cite{Jansson1998,Gustafsson2002}. In the next section we consider the problem of %parameters estimation by the gradient-based AF techniques.

We prove the following theoretical result.

\begin{lem} \label{lemma:1}
Consider state-space model~\eqref{eq:st:1}, \eqref{eq:st:2} where non-Gaussian uncertainties might arise. Similarly to the classical KF, the following formulas can be proved when the filter feedback gain obeys~\eqref{mcc:f:K}
\begin{align}
K_{k}^L &   = (P_{k|k-1}^{-1}+L_kH_k^TR_k^{-1}H_k)^{-1}L_kH^T_kR_k^{-1} \nonumber \\
& = P_{k|k}L_kH_k^TR_k^{-1} \label{mcc:K:eq2} \\
 & = P_{k|k-1}L_kH_k^T\left(R_{e,k}^L\right)^{-1}, \; R_{e,k}^L = H_kP_{k|k-1}L_kH_k^T+R_k. \label{mcc:K:eq1} \\
P_{k|k}  & \! = \left( P_{k|k-1}^{-1}+L_kH^T_kR_k^{-1}H_k\right)^{-1} \label{mcc:P:eq1} \\
 &  = (I - K_{k}^LH_k)P_{k|k-1} \label{mcc:P:eq2} \\
 &  = (I\! - \!K_{k}^LH_k)P_{k|k-1}(I  -  K_{k}^LH_kL_k)^T+K_k^LR_k\left(K_k^L\right)^T \label{mcc:P:eq3}
\end{align}
where $L_k$ is computed by formula~\eqref{mcc:f:L}.
%Similarly to the classical KF we denote $R_{e,k} = H_kP_{k|k-1}L_kH_k^T+R_k$.
\end{lem}

\begin{proof} First, we note that formula~\eqref{mcc:f:K} of the original MCC-KF, i.e. $K_k^L = (P_{k|k-1}^{-1}+L_kH_k^TR_k^{-1}H_k)^{-1}L_kH^T_kR_k^{-1}$, is obtained by a simple substitution of~\eqref{mcc:P:eq1} into~\eqref{mcc:K:eq2}.

Next, we prove the equivalence of formulas~\eqref{mcc:K:eq2} and~\eqref{mcc:K:eq1} used for the feedback gain computation $K_k^L$. Having substituted~\eqref{mcc:P:eq2} into~\eqref{mcc:K:eq2}, we arrive at
\begin{align*}
%K_k &  = P_{k|k}L_kH_k^TR_k^{-1}  \\
K_k^L & =  P_{k|k-1}L_kH_k^TR_k^{-1} - K_{k}^LH_kP_{k|k-1}L_kH_k^TR_k^{-1}.
\end{align*}
Hence,
\begin{align*}
K_k^L(I + H_kP_{k|k-1}L_kH_k^TR_k^{-1}) & =  P_{k|k-1}L_kH_k^TR_k^{-1}, \\
K_k^L(R_k + H_kP_{k|k-1}L_kH_k^T)R_k^{-1} & =  P_{k|k-1}L_kH_k^TR_k^{-1}
\end{align*}
and, finally, we get
\begin{align*}
K_k^L & =  P_{k|k-1}L_kH_k^T\left(R_k + H_kP_{k|k-1}L_kH_k^T\right)^{-1}.
\end{align*}
For the classical KF we have $K_k = P_{k|k-1}H_k^TR_{e,k}^{-1}$ where $R_{e,k} = H_kP_{k|k-1}H_k^T + R_k$. Similarly we define $R_{e,k}^{L} = H_kP_{k|k-1}L_kH_k^T+R_k$ and the formula above can be written in the following form: $K_k^L = P_{k|k-1}L_kH_k^T\left(R_{e,k}^L\right)^{-1}$ where $R_{e,k}^{L} = H_kP_{k|k-1}L_kH_k^T+R_k$. This is exactly equation~\eqref{mcc:K:eq1}.

Next, we need to prove the equivalence between formulas~\eqref{mcc:P:eq1}, \eqref{mcc:P:eq2} and~\eqref{mcc:P:eq3} for computing {\it a posteriori} error covariance matrix, $P_{k|k}$. First, taking into account the matrix inversion lemma, i.e. Sherman-Morrison-Woodbury formula~\cite{Higham2002}:
\[
\left(A+UCV \right)^{-1}  = A^{-1} - A^{-1}U \left(C^{-1}+VA^{-1}U \right)^{-1}VA^{-1},
\]
and by substituting~\eqref{mcc:K:eq1} into equation~\eqref{mcc:P:eq1} we obtain
\begin{align*}
P_{k|k} & = (P_{k|k-1}^{-1}+L_kH_k^TR_k^{-1}H_k)^{-1} = P_{k|k-1} \\
& - P_{k|k-1}L_kH_k^T \left(R_k +H_kP_{k|k-1}L_kH_k^T \right)^{-1} H_kP_{k|k-1} \\
& = P_{k|k-1} - K_k^LH_kP_{k|k-1} =  (I - K_{k}^LH_k)P_{k|k-1}.
\end{align*}
The last expression in the formula above is exactly equation~\eqref{mcc:P:eq2}. Hence, the Kalman-like formulas~\eqref{mcc:P:eq1}, \eqref{mcc:P:eq2} hold when  the feedback gain $K_k^L$ obeys~\eqref{mcc:f:K}.

Finally, we wish to prove~\eqref{mcc:P:eq3}. With trivial manipulations, formula~\eqref{mcc:P:eq2} is transformed to the form
\begin{align}
P_{k|k} & = (I - K_{k}^LH_k)P_{k|k-1} + K_k^LR_k\left(K_k^L\right)^T - K_k^LR_k\left(K_k^L\right)^T. \label{eq:help:1}
\end{align}
Next, by substituting~\eqref{mcc:P:eq2} into~\eqref{mcc:K:eq2}, we have
\begin{align}
K_k^L & =  (I - K_{k}^LH_k)P_{k|k-1}L_kH_k^TR_k^{-1}. \label{eq:help:2}
\end{align}
Taking into account the fact that $P_{k|k-1}$ is symmetric, and by substituting~\eqref{eq:help:2} into~\eqref{eq:help:1}, we derive equation~\eqref{mcc:P:eq3} as follows:
\begin{align*}
P_{k|k} & = (I - K_{k}^LH_k)P_{k|k-1} + K_k^LR_k\left(K_k^L\right)^T  \\
& - (I - K_{k}^LH_k)P_{k|k-1}L_kH_k^TR_k^{-1}R_k\left(K_k^L\right)^T \\
& =  (I - K_{k}^LH_k)(P_{k|k-1} - P_{k|k-1}L_kH_k^TK_k^T) + K_k^LR_k\left(K_k^L\right)^T \\
& =  (I - K_{k}^LH_k)P_{k|k-1}(I - K_{k}^LH_kL_k)^T + K_k^LR_k\left(K_k^L\right)^T.
\end{align*}

Hence, the algebraic equivalence between expressions~\eqref{mcc:P:eq2}, \eqref{mcc:P:eq3} is proved. This means that the Kalman-like formulas~\eqref{mcc:P:eq1}~-- \eqref{mcc:P:eq3} hold for {\it a posteriori} error covariance matrix when the feedback gain $K_k^L$ obeys~\eqref{mcc:f:K} or equivalently~\eqref{mcc:K:eq2}, \eqref{mcc:K:eq1}. This completes the proof.
\end{proof}

\setcounter{thm}{0}
\begin{rem}
It is interesting to note that there is a method-invariant form for calculating {\it a posteriori} error covariance $P_{k|k}$ in the classical KF and the MCC-KF approach. More precisely, equations~\eqref{kf:P:eq2} and~\eqref{mcc:P:eq2} are the same in their forms, except the way of computing the feedback gain; see the term $K_k$ in~\eqref{kf:P:eq2} and the term $K_k^L$ in~\eqref{mcc:P:eq2}. We stress that the other formulas for covariance calculation differ in their forms.
\end{rem}

\begin{rem}
The classical KF equation~\eqref{kf:P:eq3} is used in the original MCC-KF algorithm proposed in~\cite{izanloo2016kalman} for calculating the error covariance matrix $P_{k|k}$; see equation~\eqref{mcc:f:P}. In contrast to the previously-proposed MCC-KF version, Lemma~1 suggests to use formula~\eqref{mcc:P:eq3} instead. It involves the covariance inflation parameter $L_k$ in computing $P_{k|k}$ (which is not the case for the original MCC-KF). It is worth noting here that this inflation parameter is computed based on the MCC cost function and can serve as a scale to control information inflation of $P_{k|k}$.
\end{rem}

The theoretical result obtained in Lemma~\ref{lemma:1} suggests that equation~\eqref{mcc:P:eq3} should be used instead of~\eqref{mcc:f:P} in the MCC-KF computational scheme. More precisely, formula~\eqref{mcc:f:K} for the feedback gain computation seems to be inconsistent with the  MCC-KF error covariance calculation by~\eqref{mcc:f:P} because of the missing multiplier $L_k$ in~\eqref{mcc:f:P}; see formula~\eqref{mcc:P:eq3}. Hence, the accuracy of the MCC-KF filter might be improved with the above-proven theoretical result. To begin constructing the new improved MCC-KF technique (IMCC-KF), we first replace equation~\eqref{mcc:f:P} by~\eqref{mcc:P:eq3}. The results of numerical experiments presented in Section~5 confirm that this simple amendment  improves the state estimation accuracy of the original MCC-KF.

%Our next step in constructing the new estimator is a computational complexity reduction.
Next, we note that the MCC-KF feedback gain $K_k^L$ calculation in~\eqref{mcc:f:K} requires two $n \times n$ and one $m \times m$ matrices' inversions, because $P_{k|k-1} \in \mathbb{R}^{n \times n}$, $(P_{k|k-1}^{-1}+L_kH^T_kR_k^{-1}H_k) \in \mathbb{R}^{n \times n}$ and $R_k \in \mathbb{R}^{m \times m}$. Therefore, such MCC-KF implementation becomes impractical when the dimensions of the dynamic state and the measurement vector increase. Apart from the computation complexity issue, it is also preferable to avoid the matrix inversion operation. The latter is particularly advantageous from the numerical stability viewpoint. In our novel IMCC-KF technique, we use formula~\eqref{mcc:K:eq1} instead of~\eqref{mcc:f:K} in the feedback gain computation $K_k^L$. This modification avoids two $n \times n$ matrices' inversions and requires only one inversion of the matrix $R_{e,k}^{L} = H_kP_{k|k-1}L_kH_k^T+R_k$. The described amendment can be performed because of the algebraic equivalence of equations~\eqref{mcc:f:K} and~\eqref{mcc:K:eq2}, \eqref{mcc:K:eq1} proved in Lemma~1. Additionally, following the discussion in~\cite[p.~129]{simon2006optimal} we suggest to use computationally simpler expression~\eqref{mcc:P:eq2} than~\eqref{mcc:P:eq3} for calculating $P_{k|k}$. As proved in Lemma~1, these are mathematically equivalent and, hence, can be both utilized in the $P_{k|k}$ calculation.
In summary, all the mentioned improvements yield the new IMCC-KF estimator summarized in the form of Algorithm~1 below.

%_______________________________________
%\noindent\rule{\linewidth}{0.3pt}
\begin{codebox}
\Procname{{\bf Algorithm~1}. $\proc{IMCC-KF}$ ({\it Improved conventional version})}
\li \textsc{Initialization:} (k=0) $\hat x_{0|0} = \bar x_0$ and $P_{0|0} = \Pi_0$.
\zi \textsc{Time Update}: (k=1, \ldots, N) \Comment{\small\textsc{Priori estimation}}
\li \>$\hat x_{k|k-1}  = F_{k-1}\hat x_{k-1|k-1},$ \label{ic:mcc:p:X}
\li \>$P_{k|k-1}  = F_{k-1}P_{k-1|k-1}F_{k-1}^T+G_{k-1}Q_{k-1}G_{k-1}^T.$ \label{ic:mcc:p:P}
\end{codebox}
\begin{codebox}
\setlinenumberplus{ic:mcc:p:P}{1}
\zi \textsc{Measurement Update}:  (k=1, \ldots, N)  \Comment{\small\textsc{Posteriori estimation}}
\li \>Compute $L_k$ by MCC-KF formula~\eqref{mcc:f:L}, \label{ic:mcc:f:L}
\li \>$K_{k}^L  = P_{k|k-1}L_kH_k^T(H_kP_{k|k-1}L_kH_k^T+R_k)^{-1},$ \label{ic:mcc:f:K}
\li \>$\hat x_{k|k}  =    \hat x_{k|k-1}+K_{k}^L(z_k-H_k\hat x_{k|k-1}),$   \label{ic:mcc:f:X}
\li \>$P_{k|k}  = (I - K_{k}^LH_k)P_{k|k-1}.$ \label{ic:mcc:f:P}
\end{codebox}
%\noindent\rule{\linewidth}{0.3pt}
%  Example how to refer
%  Suppose we change line~\ref{kf:p:X} of \proc{conventional KF}
%_______________________________________

The above-presented IMCC-KF  is formulated in the so-called {\it conventional} form, i.e. Algorithm~1 recursively updates the entire matrices $P_{k|k-1}$ and $P_{k|k}$ in each iteration step of the filter. The method can be improved further by noting that any covariance matrix is symmetric and, hence, only its upper-triangular (or lower-triangular) part is to be re-calculated in each iteration step of the filter, only. To implement the mentioned modification, the square-root (SR) approach is widely used in the KF community. The resulted SR filters are inherently more stable (with respect to roundoff errors) than any {\it conventional} implementation and, hence, they are preferable for practical use; see the numerical results of ill-conditioned tests in~\cite[Chapter~6]{GrewalAndrews2015}.
%In the next section, we develop some square-root algorithms for the newly-developed IMCC-KF.

\section{New Square-Root IMCC-KF Implementations}  \label{filters}

The most popular approach for designing factored-form KF implementations (square root filters) is grounded in the covariance matrix Cholesky decomposition; see the detailed explanation in~\cite[p.~18]{GrewalAndrews2015}.  The important fact to be taken into account is that the Cholesky decomposition exists and is unique when the symmetric matrix to be decomposed is positive definite~\cite{Golub1983}. If the matrix is a positive semi-definite, then the Cholesky decomposition still exists, however, it is not unique~\cite{Higham1990}. More precisely, the Cholesky decomposition implies the factorization of a symmetric positive definite matrix $A$ in the following form: $A=(A^{1/2})^T(A^{1/2})$. Such factors can be made unique by insisting, for instance, that the factors have a triangular form (with positive diagonal elements) or to be symmetric~\cite{Park1994a}. In most applications, the triangular form is preferred. However, we may remark that sometimes it is not required and, hence, other SR filtering variants might be considered for various reasons; e.g., it saves computations in~\cite{Park1994a}.

In this paper, we use the Cholesky decomposition of a symmetric positive definite matrix $A$ in the following form: ${A=A^{T/2} A^{1/2}}$ where $A^{1/2}$ is an upper triangular matrix with positive diagonal elements. For convenience, we also write ${A^{-1/2} \equiv (A^{1/2})^{-1}}$, ${A^{-T/2} \equiv (A^{-1/2})^T}$.
The key idea of the  SR filtering strategy is a replacement of the state error covariance matrix, $P$, by its Cholesky factors and, then, re-formulation of the filtering equations in terms of these factors $P^{T/2}$ and $P^{1/2}$ only. Undoubtedly, the SR approach is not free of roundoff errors, however, it is motivated by two considerations~\cite{KaminskiBryson1971}: ``1) the product $P^{T/2}P^{1/2}$ can never be indefinite, even in the presence of roundoff errors, while roundoff errors sometimes cause the computed value of $P$ to be indefinite; 2) the numerical conditioning of $P^{1/2}$ is generally much better than that of $P$. More precisely, the condition number $K(P) = K(P^{T/2}P^{1/2}) = [K(P^{1/2})]^2$. This means that while numerical operation with $P$ may encounter difficulties when $K(P) = 10^p$, the SR filter should function until $K(P) = 10^{2p}$, i.e. with double precision.''
%Finally, the product $P = P^{T/2}P^{1/2}$ is always a symmetric matrix with nonnegative elements on the diagonal.

Furthermore, modern SR methods imply $QR$ factorization in each iteration step of the filter for updating the corresponding Cholesky  factors as follows: first,  the pre-array $A$ is built from the filter quantities that are available at the current step. Next, an orthogonal operator $\mathfrak{V}$ is applied to the pre-array in order to get an upper triangular (or lower triangular) form of the post-array $R$ such that $\mathfrak{V}A=R$. Finally, the updated filter quantities are simply read off from the post-array $R$; see~\cite[Chapter~12]{KailathSayed2000}.
%In this paper, we assume that the post-array $R$ is always a block upper triangular matrix.

Taking into account that $L_k$ in~\eqref{mcc:f:L} is a scalar value, two SR-based IMCC-KF implementations are designed, below.

%_______________________________________
%\noindent\rule{\linewidth}{0.3pt}
\begin{codebox}
\Procname{{\bf Algorithm~2}. $\proc{SR-based IMCC-KF}$ ({\it Square-root algorithm})}
\zi \textsc{Initialization:} (k=0)
\li \>Apply Cholesky decomposition: $\Pi_0 = \Pi_0^{T/2}\Pi_0^{1/2}$
\li \>Set the initial values: $\hat x_{0|0} = \bar x_0$ and $P_{0|0}^{1/2} = \Pi_0^{1/2}$. \label{sr:initial}
\zi \textsc{Time Update}: (k=1, \ldots, N) \Comment{\small\textsc{Priori estimation}}
%\li \>Apply Cholesky decomposition: $Q_{k-1} = Q_{k-1}^{T/2}Q_{k-1}^{1/2}$.
\li \>Repeat line~\ref{ic:mcc:p:X} of the \proc{IMCC-KF} to find $\hat x_{k|k-1}$.
\li \>$\mathfrak{V}
\underbrace{
\begin{bmatrix}
P_{k-1|k-1}^{1/2}F_{k-1}^T\\
Q_{k-1}^{1/2}G^T_{k-1}
\end{bmatrix}
}_{\mbox{\small Pre-array}}
  =
\underbrace{
\begin{bmatrix}
P_{k|k-1}^{1/2} \\
0
\end{bmatrix}
}_{\mbox{\small Post-array}}$.  \label{sr:mcc:p:P}
\zi \textsc{Measurement Update}:  (k=1, \ldots, N)  \Comment{\small\textsc{Posteriori estimation}}
%\li \>Apply Cholesky decomposition: $R_{k} = R_{k}^{T/2}R_{k}^{1/2}$.
\li \>Compute $L_k$ by MCC-KF formula~\eqref{mcc:f:L}.
\li \>$\mathfrak{V}
\underbrace{
\begin{bmatrix}
R_k^{1/2} & 0  \\
L_k^{1/2}P_{k|k-1}^{1/2}H_k^T & P_{k|k-1}^{1/2}
\end{bmatrix}
}_{\mbox{\small Pre-array}}
  =
\underbrace{
\begin{bmatrix}
\left(R_{e,k}^L\right)^{1/2} & \left(\bar K_{k}^L\right)^T \\
 0 & P_{k|k}^{1/2}
\end{bmatrix}
}_{\mbox{\small Post-array}}$,  \label{SR:p:2}
\li \>$\hat x_{k|k} = \hat x_{k|k-1}+L_k^{1/2}\bar K_{k}^L\left(R_{e,k}^L\right)^{-T/2}(z_k-H_k\hat x_{k|k-1}).$ \label{SR:x:2}
\end{codebox}
%_______________________________________
%\end{codebox}
%\begin{codebox}
%\setlinenumberplus{sr:mcc:f:L}{1}
%\setlinenumber{li:sr:mcc:f:L}
%_______________________________________
%  \noindent\rule{\linewidth}{0.3pt}
%  Example how to refer
%  Suppose we change line~\ref{kf:p:X} of \proc{conventional KF}
%_______________________________________

As can be seen, the Cholesky decomposition is applied only once for the error covariance matrix factorization. In fact, only the initial $\Pi_0$ is decomposed, i.e. $\Pi_0 = \Pi_0^{T/2}\Pi_0^{1/2}$. Then, the Cholesky  factors $P_{k|k-1}^{1/2}$ and $P_{k|k}^{1/2}$ are recursively updated instead of the full matrices $P_{k|k-1}$ and $P_{k|k}$. For that, stable orthogonal transformations are utilized as far as possible. In this paper, we develop the methods where $\mathfrak{V}$ is any orthogonal transformation such that the corresponding post-array is a block {\it upper triangular} matrix. Finally, we remark that, in any SR-based filter, the influence of roundoff  errors is still  present, however, the resulted error covariance $P_{k|k}=P_{k|k}^{T/2}P_{k|k}^{1/2}$ (and $P_{k|k-1}=P_{k|k-1}^{T/2}P_{k|k-1}^{1/2}$) is much more likely to be a positive definite matrix~\cite{KailathSayed2000}. In fact, this matrix product will be always a symmetric matrix with nonnegative diagonal entries.

The SR variant (Algorithm~2) is algebraically equivalent to the {\it conventional} \proc{IMCC-KF} implementation (Algorithm~1). It can be easily proved by taking into account the properties of orthogonal matrices. Indeed, from a general form of the SR-based filter iterates (i.e. $\mathfrak{V}A=R$), we obtain $R^T R = A^T\mathfrak{V}^T\mathfrak{V}A = A^TA$. Next, comparing both sides of the resulted matrix equality $A^TA=R^TR$ we derive the required formulas.  More precisely, from equation in line~\ref{sr:mcc:p:P} of the \proc{SR-based IMCC-KF} we obtain
\[F_{k-1}P_{k-1|k-1}^{T/2}P_{k-1|k-1}^{1/2}F_{k-1}^T+G_{k-1}Q_{k-1}^{T/2}Q_{k-1}^{1/2}G_{k-1}^T=P_{k|k-1}^{T/2}P_{k|k-1}^{1/2},\]
which is exactly the equation in line~\ref{ic:mcc:p:P} of the {\it conventional} \proc{IMCC-KF}.

From line~\ref{SR:p:2} of Algorithm~2 we have the following equalities:
\begin{align}
R_k^{T/2}R_k^{1/2}+H_kP_{k|k-1}^{T/2}L_kP_{k|k-1}^{1/2}H_k^T & =  \left(R_{e,k}^L\right)^{T/2}\left(R_{e,k}^L\right)^{1/2}, \label{Use:1} \\
P_{k|k-1}^{T/2}L_k^{1/2}P_{k|k-1}^{1/2}H_k^T  & = \bar K_{k}^L \left(R_{e,k}^L\right)^{1/2},  \label{Use:2} \\
P^{T/2}_{k|k}P^{1/2}_{k|k} + \bar K_k^L \left(\bar K_k^L\right)^T  & = P^{T/2}_{k|k-1}P^{1/2}_{k|k-1}. \label{Use:3}
\end{align}
 Since $L_k$ is a scalar value, we obtain $R_{e,k}^L = H_kP_{k|k-1}L_kH_k^T+R_k$ from equation~\eqref{Use:1}. Next, from~\eqref{Use:2} we have
 \begin{equation}
 \bar K_{k}^L=L_k^{1/2}P_{k|k-1}H_k^T(R_{e,k}^L)^{-1/2}, \label{eq:normK}
 \end{equation}
   which is the ``normalized'' feedback gain. According to equation~\eqref{mcc:K:eq1}, the  relation between $K_k^L$ in line~\ref{ic:mcc:f:K} of the \proc{IMCC-KF} and its ``normalized'' version $\bar K_k^L$ in line~\ref{SR:p:2} of the \proc{SR-based IMCC-KF} is the following: $K_k^L = L_k^{1/2} \bar K_k^L (R_{e,k}^L)^{-T/2}$. Thus, from equation in line~\ref{ic:mcc:f:X} of Algorithm~1 we obtain the formula for $\hat x_{k|k}$ computation in  line~\ref{SR:x:2} of Algorithm~2 as follows:
 \begin{align*}
 \hat x_{k|k} &  =  \hat x_{k|k-1}+K_{k}^L(z_k-H_k\hat x_{k|k-1}) \\
 & = \hat x_{k|k-1}+L_k^{1/2}\bar K_{k}^L\left(R_{e,k}^L\right)^{-T/2}(z_k-H_k\hat x_{k|k-1}).
 \end{align*}

Finally, taking into account that the error covariance matrix is symmetric and $L_k$ is a scalar, from~\eqref{Use:3}, \eqref{eq:normK}, we obtain
\begin{align*}
P_{k|k} & = P_{k|k-1}-\bar K_k^L \left(\bar K_k^L\right)^T \\
& = P_{k|k-1} - L_k^{1/2}P_{k|k-1}H_k^T(R_{e,k}^L)^{-1/2}(R_{e,k}^L)^{-T/2}H_kP_{k|k-1}L_k^{1/2} \\
& = P_{k|k-1} - P_{k|k-1}L_kH_k^T(R_{e,k}^L)^{-1}H_kP_{k|k-1} \\
& = P_{k|k-1} - K_k^LH_kP_{k|k-1} = (I - K_{k}^LH_k)P_{k|k-1}.
\end{align*}
The last expression in the formula above is exactly the equation in line~\ref{ic:mcc:f:P} of the {\it conventional} \proc{IMCC-KF} (Algorithm~1). This completes the proof of   algebraic equivalence between the \proc{IMCC-KF}  and its SR-based variant (Algorithm~2).

The analysis of Algorithms~1 and 2 suggests that their implementations demand one $m\times m$ matrix inversion; see lines~\ref{ic:mcc:f:K} and~\ref{SR:x:2}, respectively. However, in contrast to the {\it conventional} implementation (Algorithm~1), the SR-based variant (Algorithm~2) requires the inversion of only upper triangular matrix $\left(R_{e,k}^L\right)^{1/2}$ instead of $R_{e,k}^L$. The latter can be done by the computationally cheap backward substitution.

Eventually, Algorithm~2 can be improved further such that the new method avoids $R_{e,k}^L$ (or its  Cholesky  factor) inversion. Following~\cite{ParkKailath1995}, we develop the {\it extended} SR-based version.
%_______________________________________
%\noindent\rule{\linewidth}{0.3pt}
\begin{codebox}
\Procname{{\bf Algorithm~3}. $\proc{eSR-based IMCC-KF}$ ({\it extended SR algorithm})}
\zi \textsc{Initialization:} (k=0)
\li \>Apply Cholesky decomposition: $\Pi_0 = \Pi_0^{T/2}\Pi_0^{1/2}$
\li \>Set $P_{0|0}^{1/2} = \Pi_0^{1/2}$ and $P^{-T/2}_{0|0}\hat x_{0|0} = \Pi_0^{-T/2}\bar x_{0|0}$. \label{sr:initial1}
\zi \textsc{Time Update}: (k=1, \ldots, N) \Comment{\small\textsc{Priori estimation}}
%\li \>Apply Cholesky decomposition: $Q_{k-1} = Q_{k-1}^{T/2}Q_{k-1}^{1/2}$.
\li $\mathfrak{V}
\underbrace{
\begin{bmatrix*}[l]
P_{k-1|k-1}^{1/2}F_{k-1}^T \!\!\!& | \: P^{-T/2}_{k-1|k-1}\hat x_{k-1|k-1}\\
Q_{k-1}^{1/2}G^T_{k-1}     \!\!\!& | \: \quad 0
\end{bmatrix*}
}_{\mbox{\small Pre-array}}
\!  =\!
\underbrace{
\begin{bmatrix*}[l]
P_{k|k-1}^{1/2} \!\!\!& | \:P^{-T/2}_{k|k-1}\hat x_{k|k-1} \\
\quad 0 \!\!\!& | \: (*)
\end{bmatrix*}
}_{\mbox{\small Post-array}}$  \label{eSR:p:1}
\zi \>where $\mathfrak{V}$ is any orthogonal transformation such that the
\zi \>first block column of the post-array is upper triangular.
\li \>Read off $\left(P_{k|k-1}^{1/2}\right)$, $\left(P^{-T/2}_{k|k-1}\hat x_{k|k-1}\right)$ from the post-array.
\li \>Compute $\hat x_{k|k-1}  = \left(P_{k|k-1}^{1/2}\right)^T \left(P^{-T/2}_{k|k-1}\hat x_{k|k-1}\right).$ \label{eSR:x:1}
\end{codebox}
\begin{codebox}
\setlinenumberplus{eSR:x:1}{1}
\zi \textsc{Measurement Update}:  (k=1, \ldots, N)  \Comment{\small\textsc{Posteriori estimation}}
%\li \>Apply Cholesky decomposition: $R_{k} = R_{k}^{T/2}R_{k}^{1/2}$.
\li \>Compute $L_k$ by MCC-KF formula~\eqref{mcc:f:L}.
%\end{codebox}
%\begin{codebox}
%\setlinenumberplus{sr:mcc:f:L}{1}
%\setlinenumber{li:sr:mcc:f:L}
\li \>$
\mathfrak{V}
\begin{bmatrix*}[l]
R_k^{1/2} & 0  & | \: -R^{-T/2}_kL_k^{1/2}z_k\\
L_k^{1/2}P_{k|k-1}^{1/2}H_k^T & P_{k|k-1}^{1/2} & | \: P^{-T/2}_{k|k-1}\hat x_{k|k-1}
\end{bmatrix*}
$  \Comment{\small{\mbox{\small Pre-array}}} \label{eSR:p:2}
\zi \>$ =
\begin{bmatrix*}[l]
\left(R_{e,k}^L\right)^{1/2} & \left(\bar K_{k}^L\right)^T & | \: -\bar e_k \\
 0 & P_{k|k}^{1/2} & | \: P^{-T/2}_{k|k}\hat x_{k|k}
\end{bmatrix*}$  \Comment{\small{\mbox{\small Post-array}}}
\zi \>where $\mathfrak{V}$ is any orthogonal operator such that the first
\zi \>two block columns of the post-array is upper triangular.
\li \>Read off $\left(P_{k|k}^{1/2}\right)$ and $\left(P^{-T/2}_{k|k}\hat x_{k|k}\right)$ from the post-array.
\li \>Compute $\hat x_{k|k}  = \left( P_{k|k}^{1/2}\right)^T \left(P^{-T/2}_{k|k}\hat x_{k|k}\right).$ \label{eSR:x:2}
\end{codebox}
%  \noindent\rule{\linewidth}{0.3pt}
%  Example how to refer
%  Suppose we change line~\ref{kf:p:X} of \proc{conventional KF}
%_______________________________________

As can be seen, the state vector estimates $\hat x_{k|k-1}$ and $\hat x_{k|k}$ are computed by a simple multiplication of the blocks that are directly read off from the corresponding post-arrays; see equations in lines~\ref{eSR:x:1}, \ref{eSR:x:2} of Algorithm~3. The block $(*)$ in the post-array of the \proc{eSR-based IMCC-KF} means that these entries are of no interest in the presented filter implementation.

The novel \proc{eSR-based IMCC-KF} (Algorithm~3) is algebraically equivalent to the \proc{SR-based IMCC-KF} (Algorithm~2) and, hence, to the {\it conventional} implementation in Algorithm~1. Indeed, from line~\ref{eSR:p:1} of Algorithm~3, we have
\begin{align*}
F_{k-1}P_{k-1|k-1}^{T/2}P^{-T/2}_{k-1|k-1}\hat x_{k-1|k-1} & = P_{k|k-1}^{T/2} P^{-T/2}_{k|k-1}\hat x_{k|k-1},
\end{align*}
 which is exactly the equation in line~\ref{ic:mcc:p:X} of the \proc{IMCC-KF} (Algorithm~1). Its SR-based version (Algorithm~2) utilizes the same formula for the $\hat x_{k|k-1}$ computation, i.e. $\hat x_{k|k-1}=F_{k-1}\hat x_{k-1|k-1} $.

Next, line~\ref{eSR:p:2} in Algorithm~3 yields~\eqref{Use:1}~-- \eqref{Use:3} and
\begin{align}
\!R_k^{T/2}R^{-T/2}_kL_k^{1/2}z_k - H_k P_{k|k-1}^{T/2}L_k^{1/2}\!P^{-T/2}_{k|k-1}\hat x_{k|k-1}  &
= \left(R_{e,k}^L\right)^{T/2}\! \bar e_k, \label{Use:4} \\
P_{k|k-1}^{T/2}P^{-T/2}_{k|k-1}\hat x_{k|k-1} = P_{k|k}^{T/2}P^{-T/2}_{k|k}\hat x_{k|k} - \bar K_k^L \bar e_k.  \label{Use:5}
\end{align}
Taking into account that $L_k$ is a scalar, from~\eqref{Use:4}, we have
 \begin{equation}
 \bar e_k = \left(R_{e,k}^L\right)^{-T/2}L_k^{1/2}\left(z_k - H_k \hat x_{k|k-1} \right) = \left(R_{e,k}^L\right)^{-T/2}L_k^{1/2}e_k \label{eq:bare}
 \end{equation}
where $e_k = z_k - H_k \hat x_{k|k-1}$ is the filter residual and, hence, the quantity $\bar e_k$ is its ``normalized'' version.

From formula~\eqref{Use:5}, we get $\hat x_{k|k} = \hat x_{k|k-1} + \bar K_k^L \bar e_k$. By substituting the expressions for the ``normalized'' feedback gain~\eqref{eq:normK} and the ``normalized'' residual~\eqref{eq:bare}, we get
\begin{align*}
\hat x_{k|k} & = \hat x_{k|k-1} + \bar K_k^L \bar e_k \\
& = \hat x_{k|k-1} + \bar K_{k}^L\left(R_{e,k}^L\right)^{-T/2}L_k^{1/2}(z_k-H_k\hat x_{k|k-1}) \\
& = \hat x_{k|k-1} + L_k^{1/2}P_{k|k-1}H_k^T(R_{e,k}^L)^{-1/2} \left(R_{e,k}^L\right)^{-T/2}L_k^{1/2}e_k\\
& = \hat x_{k|k-1} + P_{k|k-1}L_kH_k^T(R_{e,k}^L)^{-1} \left(z_k - H_k \hat x_{k|k-1} \right) \\
& = \hat x_{k|k-1} + K_k^L \left(z_k - H_k \hat x_{k|k-1} \right).
\end{align*}
The second expression in the formula above is the equation in line~\ref{SR:x:2} of Algorithm~2, meanwhile the last expression coincides with line~\ref{ic:mcc:f:X} of Algorithm~1. This completes the proof of algebraic equivalence of all Algorithms~1-3.

\section{Numerical Experiments}

To fulfil a fair comparative study of the newly-developed methods and the original MCC-KF proposed in~\cite{izanloo2016kalman}, we consider the test problem in the cited paper and the accompanying MATLAB codes, which are freely available in~\cite{izanloo2016codes}. To provide the same experimental conditions for all methods under examination, we incorporate our IMCC-KF algorithms into the cited codes, where the original MCC-KF  technique has been already implemented.

\begin{table*}
\caption{The RMSE errors and the average CPU time (s) in the presence of shot and Gaussian mixture noises in Example~\ref{ex:1}, respectively. $M=100$.} \label{table:1}
\centering
{\small
\begin{tabular}{l||cccc|c||cccc|c|c}
\hline
Filter &
\multicolumn{5}{c||}{Case 1: Shot noise} &
\multicolumn{5}{c|}{Case 2: Gaussian mixture noise} & CPU \\
\cline{2-11}
\cline{2-11}
Implementation  & \multicolumn{4}{c|}{RMSE$_x$}  &  &   \multicolumn{4}{c|}{RMSE$_x$} &  &  \\
 & $x_1$ & $x_2$ & $x_3$ & $x_4$ & $\|\mbox{RMSE}_x\|_2$  & $x_1$ & $x_2$ & $x_3$ & $x_4$ & $\|\mbox{RMSE}_x\|_2$ &  \\
\hline
\hline
MCC-KF: Eqs~\eqref{mcc:p:X}-\eqref{mcc:f:P}
& 2.215    &  2.228   &  0.912  &    0.918  & 3.397  & 2.904 & 2.914 & 1.649 & 1.648 & 4.728 & 0.0307 \\
MCC-KF: Eqs~\eqref{mcc:p:X}-\eqref{mcc:f:X}, \eqref{mcc:P:eq3}
& 2.190    &   2.204   &  0.908  &   0.913 &  3.363 & 2.890  & 2.900 & 1.660 & 1.658 & 4.718 & 0.0316\\
\hline
%\cline{2-11}
\proc{IMCC-KF}         & 2.190  & 2.204 &  0.908  & 0.913 & 3.363 &   2.890  & 2.900 & 1.660 & 1.658 & 4.718  & 0.0140 \\
\proc{SR-based IMCC-KF} & 2.190  & 2.204 &  0.908  & 0.913 & 3.363 &  2.890  & 2.900 & 1.660 & 1.658 & 4.718  & 0.0121 \\
\proc{eSR-based IMCC-KF} & 2.190  & 2.204 &  0.908  & 0.913 & 3.363 &  2.890  & 2.900 & 1.660 & 1.658 & 4.718 & 0.0123 \\
\hline
\end{tabular}
}
\end{table*}

\begin{exmp} \label{ex:1}
 Consider a land vehicle dynamic measured in time intervals of $3$ seconds, i.e. $\Delta t = 3$ sec.,  with the heading angle $\psi$ (which is set to $\psi = 60$ deg.), as follows:
\begin{align*}
x_k & =
\begin{bmatrix}
1 & 0 & \Delta t & 0\\
0 & 1 & 0 & \Delta t \\
0& 0 & 1 & 0\\
0 & 0 & 0 & 1
\end{bmatrix}
x_{k-1}
+
\begin{bmatrix}
0 \\
0\\
\Delta t  \sin \psi \\
\Delta t  \cos \psi
\end{bmatrix}
u_{k-1} + w_{k-1}, \\
z_k & =
\begin{bmatrix}
1 & 0 & 0 & 0\\
0 & 1 & 0 & 0
\end{bmatrix}
x_k + v_k
\end{align*}
with the initial conditions being $\bar x_0 = [1, 1, 0, 0]^T$ and $\Pi_0 = diag\{[4, 4, 3, 3]\}$.
Additionally, two cases of the process and measurement noises are employed, below.

{\bf Case 1.} All entries of $w_k$ and $v_k$ are comprised of Gaussian noise plus shot noise, i.e. $w_k  = {\cal N}(\mu_x,Q)+\mbox{\tt Shot noise}$, and
$v_k = {\cal N}(\mu_z,R)+\mbox{\tt Shot noise}$. The means and covariances are taken to be $\mu_x = 0$, $\mu_z = 0$ and $Q = diag\{[0.1, 0.1, 0.1, 0.1]\}$, $R = diag\{[0.1, 0.1]\}$, respectively.

{\bf Case 2.} The $w_k$ and $v_k$ are Gaussian mixture noises. The $w_k$ is a mixture of ${\cal N} (\mu_{x1}, Q_1)$ and ${\cal N} (\mu_{x2},Q_2)$. The $v_k$ is a mixture of ${\cal N} (\mu_{z1}, R_1)$ and ${\cal N} (\mu_{z2},R_2)$. The means and covariances are taken to be $\mu_{x1} = [-3, -3, -3, -3]^T$, $\mu_{x2} = [2, 2, 2, 2]^T$,  $\mu_{z1} = [2, 2]^T$, $\mu_{z2} = [-2, -2]^T$, $Q_1=Q_2=Q$ and $R_1=R_2=R$.
\end{exmp}
\vspace{-0.3cm}

In our comparative study, the following methods are tested: 1) the original MCC-KF scheme given by equations~\eqref{mcc:p:X}~-- \eqref{mcc:f:P};
2) the original MCC-KF where only one equation is corrected, i.e. formula~\eqref{mcc:f:P} is replaced by~\eqref{mcc:P:eq3};
3) the improved method, i.e. the \proc{IMCC-KF} (Algorithm~1); 4) the \proc{SR-based IMCC-KF} (Algorithm~2); and 5) the \proc{eSR-based IMCC-KF} (Algorithm~3).

We repeat the set of numerical experiments from~\cite{izanloo2016kalman}. More precisely, the stochastic model in Example~\ref{ex:1} is simulated for $k = 1, \ldots, 300$ to generate the measurement history. For that, the shot noise (case~1) and Gaussian mixture noise (case~2) are treated separately. They are implemented exactly as in~\cite{izanloo2016codes}.
%Following the implementation in~\cite{izanloo2016codes}, the Gaussian mixture noises are generated as follows: in each time instance $k$, one generates five random samples from ${\cal N} (\mu_{x1}, Q_1)$ and ${\cal N} (\mu_{x2},Q_2)$, respectively. Next, we sample randomly an integer from the uniform distribution in the interval $[1,10]$ for choosing a particular noise $w_k$ realization (out of the above ten noise vectors). The measurement noise $v_k$ is determined in the same way except that samples are taken from ${\cal N} (\mu_{z1}, R_1)$ and ${\cal N} (\mu_{z2},R_2)$. Then, we randomly fix one of these ten realizations.
Next, the inverse problem is solved, i.e. the optimal dynamic state estimate is computed by each filtering technique under examination. We repeat the experiment $M=100$ times. To judge the quality of the estimators, the root mean square error (RMSE) in each component of the state vector averaged over $100$ Monte Carlo runs is computed as follows:
\begin{equation} \label{eq:RMSEx}
\mbox{\rm RMSE}_{x_i}=\sqrt{\frac{1}{MN}\sum \limits_{j=1}^{M} \sum \limits_{k=1}^N \left(x_{i, exact}^j(t_k) - \hat x_{i, k|k}^j\right)^2}
\end{equation}
where $M=100$ is the number of Monte-Carlo trials, $N=300$ is the number of discrete time points, the $x_{i, exact}^j(t_k)$ and $\hat x_{i, k|k}^j$ are the $i$-th entry of the ``true'' state vector (simulated) and its estimated value obtained in the $j$-th Monte-Carlo trial, respectively. The results of this set of numerical experiments are summarized in Table~\ref{table:1}, where we report the $\mbox{\rm RMSE}_{x_i}$, $i=1, \ldots, 4$, the $\|\mbox{RMSE}_x\|_2$ and the CPU time (s) averaged over $M=100$ Monte-Carlo simulations. All methods were implemented in the same precision (64-bit
floating point) in MATLAB running on a conventional PC with processor \verb"Intel(R) Core(TM) i5-2410M CPU 2.30 GHz" and with $4$ GB of installed memory (RAM).

Table~\ref{table:1} allows for the following conclusions. First, having compared the first two rows, we observe that a simple replacement of formula~\eqref{mcc:f:P} in the previously derived MCC-KF by equation~\eqref{mcc:P:eq3} in accordance with Lemma~1  enhances the filter estimation accuracy. Recall that the difference between these two versions is only in the error covariance $P_{k|k}$ computation. Although the previously proposed MCC-KF~\eqref{mcc:p:X}-\eqref{mcc:f:P} is formulated in the symmetric Joseph stabilized form~\eqref{mcc:f:P}, the results of numerical experiments suggest that appearance of multiplier $L_k$ in~\eqref{mcc:f:P} according to Lemma~1 improves the estimation accuracy in all entries of the vector $\hat x_{k|k}$ in the shot noise (case~1) and in half of the state vector components in the Gaussian mixture noise (case 2). The aggregated $\|\mbox{RMSE}_x\|_2$ quantity is less for both noise cases when the MCC-KF methodology is implemented via recursion~\eqref{mcc:p:X}-\eqref{mcc:f:X}, \eqref{mcc:P:eq3}.  In summary, the previously developed MCC-KF given by equation~\eqref{mcc:p:X}-\eqref{mcc:f:P} is less accurate compared to its novel implementation presented in Section~3, where $L_k$ is taking into account. Concerning the computational complexity, the difference between these two implementations is imperceptible and, hence, their average CPU time is the same.

Second, Table~\ref{table:1} says that Algorithms~1-3 maintain the same estimation accuracy in both noise case scenarios. This conclusion holds for all entries of the state vector to be estimated and for the aggregated $\|\mbox{RMSE}_x\|_2$ values. Hence, our theoretical expectations are realized. Indeed, the mentioned finding was anticipated, since Algorithms~1-3 are proved to be algebraically equivalent in Section~4. It is also in line with the main theoretical result proved in Lemma~1. In other words, the outcome of our numerical experiments substantiates the theoretical derivations of Lemma~1 in practice.

Concerning the computational complexity of the original MCC-KF and the novel IMCC-KF algorithm, we observe the following: the CPU time is two times lower for the IMCC-KF than for the MCC-KF. It was expected since the IMCC-KF does not require two $n \times n$ matrices' inversion for calculating the gain in contrast to the original MCC-KF. Besides, the new IMCC-KF uses a computationally simpler expression for $P_{k|k}$ computation than the original MCC-KF.

Next, consider the SR algorithms. In general, they are more computationally expensive, but inherently more stable and reliable than any {\it conventional} implementation, because of the use of numerically stable orthogonal transformations in each iteration step of the filter. We observe that for  our low-dimensional problem in Example~1 ($n=4$, $m=2$), the average CPU time of the SR IMCC-KF is almost the same as in the corresponding conventional implementation, i.e. in the IMCC-KF.

Finally, recall that the SR-based variants (Algorithms~2, 3) are numerically more robust (with respect to roundoff errors) than the {\it conventional} implementation (Algorithm~1). However, Example~1 does not suite for examination of numerical insights of the proposed computational schemes. Further, we equip Example~1 with an ill-conditioned measurement model, as it is often done in the KF literature while investigating the filter numerical stability issues; e.g., see~\cite{GrewalAndrews2015}.

\begin{table*}
\caption{
The effect of roundoff errors on the computed $\|\mbox{RMSE}_x\|_2$ in the presence of shot and Gaussian mixture noises in Example~2, respectively. $M=100$.
} \label{table:2}
\centering
{\small
\begin{tabular}{l||cccccc||cccccc}
\hline
Filter &
\multicolumn{6}{c||}{Case 1: Shot noise,  $\delta \to \epsilon_{roundoff}$} &
\multicolumn{6}{|c}{Case 2: Gaussian mixture noise,  $\delta \to \epsilon_{roundoff}$} \\
\cline{2-13}
Implementation  & $10^{-2}$ & $10^{-3}$ & $10^{-4}$ & $10^{-5}$ & $10^{-6}$ & $10^{-7}$ & $10^{-2}$ & $10^{-3}$ & $10^{-4}$ & $10^{-5}$ & $10^{-6}$ & $10^{-7}$ \\
\hline
\hline
MCC-KF: Eqs~\eqref{mcc:p:X}~-\eqref{mcc:f:P}
                         & 40.26 & 43.75   &  39.39  &   40.372  & 44.505  & \verb"NaN" &  43.38 & 41.72 &  39.70 &  39.82 &  \verb"NaN" & \verb"NaN" \\
MCC-KF: Eqs~\eqref{mcc:p:X}~-\eqref{mcc:f:X}, \eqref{mcc:P:eq3}
                         & 38.81 & 40.99   &  37.07  &   38.986 & \verb"NaN" & \verb"NaN" &  41.02 &   40.85  &  37.14 &   37.72 & \verb"NaN" & \verb"NaN" \\
\hline
\proc{IMCC-KF}           & 38.81 & 40.99 &   37.07 &     38.985 &  39.138 & \verb"NaN" & 41.02 & 40.85  & 37.14 & 37.68 & 42.19 & \verb"NaN" \\
\proc{SR-based IMCC-KF}  & 38.81 & 40.99 &   37.07 &     38.985 &  39.133 &  45.72     & 41.02 & 40.85   & 37.14 & 37.68 &  42.13 &  38.02 \\
\proc{eSR-based IMCC-KF} & 38.81 & 40.99 &   37.07 &     38.985 &  39.133 &  45.72     & 41.02 & 40.85   & 37.14 & 37.68 &  42.13 &  38.02 \\
\hline
\end{tabular}
}
\end{table*}

\begin{exmp} \label{ex:2}
Consider a land vehicle dynamic from Example~\ref{ex:1}
\begin{align*}
z_k & =
\begin{bmatrix}
1 & 1 & 1 & 1\\
1 & 1 & 1 & 1+\delta
\end{bmatrix}
x_k + v_k, \; v_k \sim {\cal N}(0,R)
\end{align*}
where $R = diag\{[\delta^2, \delta^2]\}$. To simulate roundoff we assume that
$\delta^2<\epsilon_{roundoff}$, but $\delta>\epsilon_{roundoff}$
where $\epsilon_{roundoff}$ denotes the unit roundoff
error\footnote{Computer roundoff for floating-point arithmetic is
often characterized by a single parameter $\epsilon_{roundoff}$,
defined in different sources as the largest number such that either
$1+\epsilon_{roundoff} = 1$ or $1+\epsilon_{roundoff}/2 = 1$ in
machine precision. }. Again, two cases of the process and measurement noises are employed.

{\bf Case 1.} All entries of $w_k$ are comprised of Gaussian noise plus shot noise, i.e. $w_k = {\cal N}(\mu_x,Q)+\mbox{\tt Shot noise}$ with $\mu_x = 0$ and $Q = diag\{[0.1, 0.1, 0.1, 0.1]\}$.

{\bf Case 2.} It is the same as in Example~\ref{ex:1}.
\end{exmp}

We repeat the experiments described above for $100$ Monte Carlo runs and
various ill-conditioning parameter values $\delta$. The source of the difficulty is in the matrix $R_{e,k}^L$ inversion. More precisely, we remark that although ${\rm rank}\: H =2$, the matrix $R_{e,k}^L$ becomes severely ill-conditioned as $\delta \to \epsilon_{roundoff}$, i.e. to machine precision limit. In summary, we provide the following set of numerical experiments. For each value of the parameter $\delta$, we solve the state estimation problem as it is done above in  Example~1. Then, we summarize the aggregated values $\|\mbox{RMSE}_x\|_2$  in Table~2 for each filter implementation under examination. It should be also stressed that the shot noise (case~1) is added only in the process equation in this example. Besides, following the MCC-KF implementation in~\cite{izanloo2016codes}, the noise covariances should be adapted to the sample values when the noises are generated. However, in our numerical experiments we do not adjust the covariance $R$ to the corresponding sample value.

%The results summarized in Table~2 allow the numerical behaviour of each algorithm to be explored in line with the growing ill-conditioning.

%_____________________________
Having analyzed the obtained results presented in Table~2, we can explore the numerical behaviour of each algorithm while growing problem ill-conditioning. We again conclude that for both noise scenarios the previously developed MCC-KF given by equations~\eqref{mcc:p:X}-\eqref{mcc:f:P} is less accurate compared with the implementation via formulas~\eqref{mcc:p:X}-\eqref{mcc:f:X}, \eqref{mcc:P:eq3}, although equation~\eqref{mcc:f:P} is formulated in the symmetric stabilized form. Both mentioned algorithms belong to the class of {\it conventional} implementations. From the second panel of Table~2, we observe their fast degradation while $\delta \to \epsilon_{roundoff}$, the machine precision limit. Indeed, already for $\delta=10^{-6}$ they provide essentially no correct digits in the computed state estimate for Gaussian mixture noise, since the  $\|\mbox{RMSE}_x\|_2$  value is 'NaN'; see the first two rows in Table~\ref{table:2}. In MatLab, the term 'NaN' stands for 'Not a Number' that actually means the failure of numerical method. The interesting fact is that for shot noise (case~1) the symmetric Joseph stabilized form~\eqref{mcc:f:P} seems to be more robust to roundoff errors compared to the implementation via equation~\eqref{mcc:P:eq3}, which is not symmetric because of the multiplier $L_k$. Hence, we can conclude that the original MCC-KF~\eqref{mcc:p:X}-\eqref{mcc:f:P} is less accurate compared with the implementation via formulas~\eqref{mcc:p:X}-\eqref{mcc:f:X}, \eqref{mcc:P:eq3}, but seems to be more robust to roundoff errors because of the symmetric Joseph stabilized form~\eqref{mcc:f:P}.
%_____________________________

Next, we study the new IMCC-KF (Algorithm~1) that is also of {\it conventional} (non-square-root) type. We observe that the original MCC-KF based on the Joseph stabilized implementation~\eqref{mcc:f:P} is less accurate and less robust than Algorithm~1 in the both noise cases under consideration. This finding is reasonable if we recall that the new improved method (Algorithm~1) avoids two $n\times n$ matrices' inversions compared to the original MCC-KF via~\eqref{mcc:p:X}-\eqref{mcc:f:P} and the algorithm via recursion~\eqref{mcc:p:X}-\eqref{mcc:f:X}, \eqref{mcc:P:eq3}; see the discussion in Section~3.

Finally, the last two rows of Table~2 are considered. It can be seen that the SR-based algorithms (Algorithms~2, 3) outperform any {\it conventional} implementation in the estimation accuracy in the both noise case scenarios. Indeed, they degrade more slowly than any other examined counterpart as $\delta \to \epsilon_{roundoff}$. Thus, the outcome of these numerical experiments substantiates their inherent numerical stability with respect to roundoff errors.
%the algebraic equivalence of the newly-developed SR-based algorithms and also

%In summary, as ${\delta \to \epsilon_{\text{roundoff}}}$, which corresponds to growing ill-conditioning, the new improved version developed in this paper works more accurate than the previously developed MCC-KF counterpart for both noise cases

\section{Concluding Remarks}

In this paper, the recently proposed MCC-KF technique is revised. As a result, the new improved estimator and two its SR-based variants are developed. Although all new algorithms are algebraically equivalent, the results of the numerical experiments suggest that the SR-based filters provide the best estimation quality when solving ill-conditioned state estimation problem in the presence of non-Gaussian noise. The elaborated state estimation techniques are planned to be extended to state estimation of continuous-time nonlinear stochastic systems via the accurate extended continuous-discrete KF approach presented recently in~\cite{KuKu14IEEE_TAC,KuKu15EJCON,KuKu16IEEE_TSP,KuKu16SISCI}.

\section*{Acknowledgments}
The author thanks the support of Portuguese National Fund ({\it Funda\c{c}\~{a}o para a
Ci\^{e}ncia e a Tecnologia}) within the scope of project UID/Multi/04621/2013.

\section*{References}
%\bibliographystyle{elsarticle-num}
%\bibliography{BibTex_Library/books,%
%              BibTex_Library/list_MVKulikova,%
%              BibTex_Library/list_Tsyganova,%
%              BibTex_Library/list_identification,%
%              BibTex_Library/filters,%
%              BibTex_Library/list_ML,%
%              BibTex_Library/SVD,%
%              BibTex_Library/nongaussian,%
%              BibTex_Library/list_linalg}

\end{document}